\documentclass[conference]{IEEEtran}

\usepackage{hyperref}

\usepackage{nameref}
\usepackage{cite}
\def\BibTeX{{\rm B\kern-.05em{\sc i\kern-.025em b}\kern-.08em
		T\kern-.1667em\lower.7ex\hbox{E}\kern-.125emX}}

\usepackage{amsmath,amssymb,amsfonts,amssymb,mathabx}
\usepackage{graphicx}
\usepackage{textcomp}
\usepackage{enumitem} 
\usepackage{float} 
\usepackage{tikz}
\usepackage{subcaption}
\usetikzlibrary{arrows, chains, calc, positioning, shapes.multipart, fit}

\usepackage[ruled, lined, linesnumbered, commentsnumbered, longend]{algorithm2e}
\usepackage{xcolor}
\SetKwProg{Fn}{Function}{ is}{end}

\SetCommentSty{mycommfont}

\usepackage{amsthm}
\newtheorem{theorem}{Theorem}

\newcommand{\under}{\underline}

\begin{document}

\title{bsort: A theoretically efficient non-comparison-based sorting algorithm for integer and floating-point numbers}

\author{
	\IEEEauthorblockN{Benjamín Guzmán\IEEEauthorrefmark{1}}
	\IEEEauthorblockA{
		Independent researcher\\
		Mexico City, Mexico \\
		bg@benjaminguzman.dev\\
		\IEEEauthorrefmark{1}Work performed prior to joining Amazon.
	}
}

\maketitle

\begin{abstract}
This paper presents bsort, a non-comparison-based sorting algorithm for signed and unsigned integers, and floating-point values. The algorithm unifies these cases through an approach derived from binary quicksort, achieving O(wn) runtime asymptotic behavior and O(w) auxiliary space, where w is the element word size. This algorithm is highly efficient for data types with small word sizes, where empirical analysis exhibits performance competitive with highly optimized hybrid algorithms from popular libraries.
\end{abstract}

\begin{IEEEkeywords}
sorting algorithm, bitwise operations, in-place sorting, floating-point values, non-comparison sorting
\end{IEEEkeywords}

\section{Introduction}
Sorting is a fundamental problem in computer science that involves arranging an input sequence $(v_1, v_2, ..., v_n)$ into a monotonically non-decreasing or non-increasing sequence. Algorithms that solve this problem by comparing pairs of elements are known as comparison-based sorts. It is a well-established result that any such algorithm has a worst-case time complexity lower bound of $\Omega (n\ log(n))$~\cite{itb-sorting}.

Non-comparison-based algorithms, such as radix sort, offer a compelling alternative to comparison-based methods, with linear-time performance. In-place variants that operate at the bit-level, such as binary quicksort~\cite{bentley1997fast}, provide a memory-efficient solution but are limited to sorting equally-signed integer values. Extensions to handle sequences with both positive and negative integer elements~\cite{fastbit-radix-sort} and separate out-of-place methods have been extended to handle floating-point values~\cite{radix-revisited}. This paper aims to present a single, unified algorithm that is simultaneously in-place, and capable of sorting signed and unsigned integers, and floating-point values.

\section{Binary quicksort algorithm}

The binary quicksort algorithm is a divide-and-conquer algorithm which splits the array $a = [a_1, a_2, ..., a_n]$ based on the bits of its elements. The process is applied recursively, starting from the most significant bit (MSB).

In a single partitioning step, the array is split into two sub-arrays, $a'$ and $a''$, which satisfy the following invariants:

\begin{enumerate}[label=(\roman*)]
	\item $a_i \wedge m \neq 0\ \forall a_i \in a'$
	\item $a_j \wedge m = 0\ \forall a_j \in a''$
\end{enumerate}

Where $\wedge$ is the bitwise AND operator, and $m$ is a single-bit mask used in the current iteration which is defined by $m = 1 \ll (s - i)$ where $\ll$ is the left-shift bitwise operator, $s \in \mathbb{Z^+}$ is the number of bits every element in $a$ has, and $i \in \mathbb{Z}: 1 \leq i \leq s$ the current iteration, starting at 1. If $a_1$ precedes $a_2$ or vice versa depends on the desired sort direction (ascending or descending). After partitioning on one bit, the algorithm recursively sorts these sub-arrays using the mask for the next significant bit. This continues until all $s$ bits have been processed, at which point the array is fully sorted.

For example, assuming $a = [5, 7, 1, 6, 3, 4, 0]$, the entire tracing of the binary quicksort algorithm is illustrated in Fig.~\ref{fig:binary-quicksort-tracing}:

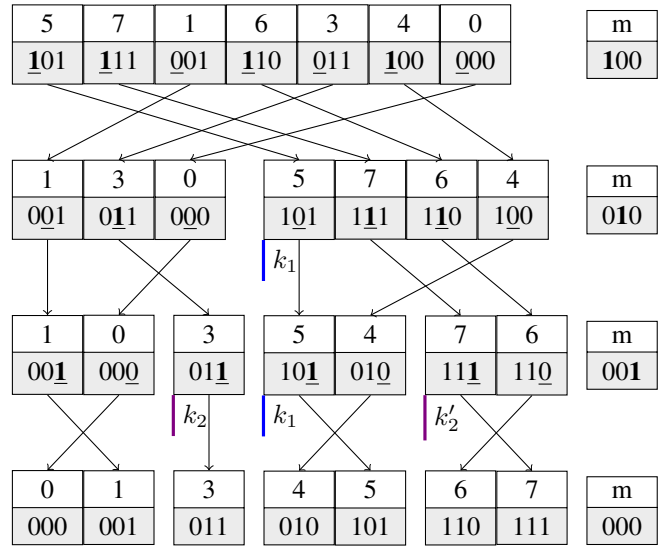
\begin{figure}[H]
	\begin{tikzpicture}[
elem/.style={rectangle split, rectangle split parts=2, 
						draw, align=center, text width=7mm,
						rectangle split part fill={white,gray!15}}
]
\node[elem]                     (01-5) 	{5 \nodepart{two} \under{\textbf{1}}01 };
\node[elem, right=0cm of 01-5]  (01-7)	{7 \nodepart{two} \under{\textbf{1}}11 };
\node[elem, right=0cm of 01-7]  (01-1)  {1 \nodepart{two} \under{0}01 };
\node[elem, right=0cm of 01-1]  (01-6)  {6 \nodepart{two} \under{\textbf{1}}10 };
\node[elem, right=0cm of 01-6]  (01-3)  {3 \nodepart{two} \under{0}11 };
\node[elem, right=0cm of 01-3]  (01-4)  {4 \nodepart{two} \under{\textbf{1}}00 };
\node[elem, right=0cm of 01-4]  (01-0)  {0 \nodepart{two} \under{0}00 };
		
\node[elem]	(01-m)	[right=1cm of 01-0]	{m \nodepart{two} \textbf{1}00 };

\node[elem, below=of 01-5]          (02-1)	{1 \nodepart{two} 0\under{0}1 };
\node[elem, right=0cm of 02-1]	    (02-3)	{3 \nodepart{two} 0\under{\textbf{1}}1 };
\node[elem, right=0cm of 02-3]      (02-0)	{0 \nodepart{two} 0\under{0}0 };
\node[elem, right=0.5cm of 02-0]    (02-5)	{5 \nodepart{two} 1\under{0}1 };
\node[elem, right=0cm of 02-5]	    (02-7)	{7 \nodepart{two} 1\under{\textbf{1}}1 };
\node[elem, right=0cm of 02-7]	    (02-6)	{6 \nodepart{two} 1\under{\textbf{1}}0 };
\node[elem, right=0cm of 02-6]    (02-4)	{4 \nodepart{two} 1\under{0}0 };

\node (k1) [text width=7mm, below=0cm of 02-5] {$k_1$};
\draw[very thick, blue] (k1.south west) |- (k1.north west);

\node[elem, below=1.15cm of 01-m]	(02-m)	{m \nodepart{two} 0\textbf{1}0 };

\draw[->] (01-5.south) -- (02-5.north);
\draw[->] (01-6.south) -- (02-6.north);
\draw[->] (01-1.south) -- (02-1.north);
\draw[->] (01-7.south) -- (02-7.north);
\draw[->] (01-3.south) -- (02-3.north);
\draw[->] (01-4.south) -- (02-4.north);
\draw[->] (01-0.south) -- (02-0.north);

\node[elem, below=of 02-1]          (03-1)	{1 \nodepart{two} 00\under{\textbf{1}} };
\node[elem, right=0cm of 03-1]      (03-0)	{0 \nodepart{two} 00\under{0} };
\node[elem, right=0.25cm of 03-0]   (03-3)	{3 \nodepart{two} 01\under{\textbf{1}} };
\node[elem, below=of 02-5]          (03-5)	{5 \nodepart{two} 10\under{\textbf{1}} };
\node[elem, right=0cm of 03-5]      (03-4)	{4 \nodepart{two} 01\under{0} };
\node[elem, right=0.25cm of 03-4]   (03-7)	{7 \nodepart{two} 11\under{\textbf{1}} };
\node[elem, right=0cm of 03-7]      (03-6)	{6 \nodepart{two} 11\under{0} };

\node (k1) [text width=7mm, below=0cm of 03-5] {$k_1$};
\draw[very thick, blue] (k1.south west) |- (k1.north west);

\node (k2) [text width=7mm, below=0cm of 03-3] {$k_2$};
\draw[very thick, violet] (k2.south west) |- (k2.north west);

\node (k2') [text width=7mm, below=0cm of 03-7] {$k_2'$};
\draw[very thick, violet] (k2'.south west) |- (k2'.north west);

\node[elem, below=1.15cm of 02-m]	(03-m)    {m \nodepart{two} 00\textbf{1} };

\draw[->] (02-5.south) -- (03-5.north);
\draw[->] (02-6.south) -- (03-6.north);
\draw[->] (02-1.south) -- (03-1.north);
\draw[->] (02-7.south) -- (03-7.north);
\draw[->] (02-3.south) -- (03-3.north);
\draw[->] (02-4.south) -- (03-4.north);
\draw[->] (02-0.south) -- (03-0.north);

\node[elem, below=of 03-1]      (04-0)  {0 \nodepart{two} 000 };
\node[elem, right=0cm of 04-0]  (04-1)  {1 \nodepart{two} 001 };
\node[elem, below=of 03-3]      (04-3)  {3 \nodepart{two} 011 };
\node[elem, below=of 03-5]      (04-4)  {4 \nodepart{two} 010 };
\node[elem, right=0cm of 04-4]  (04-5)  {5 \nodepart{two} 101 };
\node[elem, below=of 03-7]      (04-6)  {6 \nodepart{two} 110 };
\node[elem, below=of 03-6]  (04-7)  {7 \nodepart{two} 111 };

\node[elem, below=1.15cm of 03-m]   (04-m)  {m \nodepart{two} 000 };

\draw[->] (03-5.south) -- (04-5.north);
\draw[->] (03-6.south) -- (04-6.north);
\draw[->] (03-1.south) -- (04-1.north);
\draw[->] (03-7.south) -- (04-7.north);
\draw[->] (03-3.south) -- (04-3.north);
\draw[->] (03-4.south) -- (04-4.north);
\draw[->] (03-0.south) -- (04-0.north);
	\end{tikzpicture}
	\caption{Tracing of binary quicksort algorithm}
	\label{fig:binary-quicksort-tracing}
\end{figure}

\textit{bsort} uses a modified version of the \textit{binary quicksort} algorithm as its core sub-procedure, which is described in Algorithm~1, and Algorithm~2 shows how this sub-procedure is used recursively to sort equally-signed integers.

\vfill

\begin{algorithm}
	\caption{bsort sub-procedure to partition an array}
	\SetKwInOut{Input}{Input}
	\SetKwInOut{Output}{Output}
	
	\Fn{singlePassBSort(a, m, ps, pe, asc)}{
		\Input{$a$ - The array of numbers to be sorted; $a_i \in a\ \forall\ i \in \mathbb{Z}: 0 \leq i < |a|$ \newline
			$m$ - the current mask; $\exists k \in \mathbb{N}_0: m = 2^k $\newline
			$ps, pe$ - the indices at which the partition to be sorted starts (inclusive) and ends (exclusive); $ps, pe \in \mathbb{Z}: 0 \leq ps \leq pe < |a|$\newline
			$asc$ - boolean indicating whether the order is ascending
		}
		\Output{An index $k \in [ps, pe]$ such that $a_i \wedge m = 0, a_j \wedge m = m\ \forall i \in [ps, k), j \in [k, pe]$ if and only if $asc$ is true, or $a_i \wedge m = m, a_j \wedge m = 0$ otherwise}
		
		$k := ps$

        \For{$i \leftarrow ps$ \KwTo $pe$}{
            \If{$(asc \land \neg(a_i \wedge m)) \lor (\neg asc \land (a_i \wedge m))$}{
				$swap(a_i, a_k)$

				$k \leftarrow k + 1$
			}
        }

%
%
%
%
		
		\KwRet{k}
	}
	\label{alg:single-pass-bsort}
\end{algorithm}

\begin{algorithm}
	\caption{Binary quicksort for equally-signed integers}
	\SetKwInOut{Input}{Input}
	\SetKwInOut{Output}{Output}
	
	\Fn{binaryQuickSort(a, m, ps, pe, asc)}{
		\Input{Same as $singlePassBSort$}
		\Output{The same array $a$ but with the partition starting at $ps$ and ending at $pe$ (exclusive) now sorted in ascending or descending order depending of $asc$}
		
		$k:=singlePassBSort(a, m, ps, pe, asc)$
		
		$m' := m \gg 1$
		
		\If{$m' = 0$}{
			\KwRet{$a$}
		}
		
		$binaryQuickSort(a, m', ps, k, asc)$
		
		$binaryQuickSort(a, m', k, pe, asc)$
		
		\KwRet{a}
	}
	\label{alg:binary-quicksort}
\end{algorithm}

\pagebreak

\section{bsort}

The bitwise partitioning mechanism of the standard binary quicksort algorithm is predicated on the lexicographical ordering of unsigned integer representations. Consequently, it is not directly applicable to data types such as signed integers or floating-point numbers, which do not adhere to this ordering principle.

\textit{Proof}. Consider an array $a$ containing two signed integers, $a=[x, y]$, where $x < 0$ and $y \ge 0$. Let $w$ be the word size in bits. In the standard two's complement representation, the MSB of $x$ is 1, while the MSB of $y$ is 0.

The algorithm's initial partitioning pass, as defined by $singlePassBSort$, operates on the MSB and partitions the array based on the bitwise AND with the mask $m = 1 \ll (w - 1)$

\begin{enumerate}
	\item For element $x$, $x \wedge m = m$, since its MSB is 1.
	\item For element $y$, $y \wedge m = 0$, since its MSB is 0.
\end{enumerate}

According to the logic in $singlePassBSort$, for an ascending sort, elements yielding a non-zero result are placed in the second partition (the ``greater'' partition), while those yielding zero are placed in the first. Therefore, the partitioning step will incorrectly place $x$ after $y$. This violates the monotonicity invariant required for a sorted sequence, proving the algorithm fails for mixed-sign integers.~$\blacksquare$

This same logic extends to floating-point numbers. Under the IEEE-754 standard, the MSB is also a sign bit. A similar bitwise comparison will incorrectly order negative and positive values for the same reason outlined above. Furthermore, the lexicographical order of the remaining bits does not correspond directly to the numerical magnitude of floating-point values, creating additional sorting failures even within partitions of the same sign.

\subsection{bsort for signed integers}
\label{subsection:bsort for signed integers}

The preceding analysis demonstrates that the initial call to $singlePassBSort$ for signed numbers inverts the relative order of the negative and non-negative partitions. For example, in an ascending sort, this erroneously places all negative numbers after all non-negative numbers. To rectify this, we simply invert the sort direction for the initial pass of the algorithm. Algorithm 3 formalizes this procedure.

Throughout this paper, $a \gg b$ and $a \ggg b$  denote arithmetic and logical right-shift operations, respectively, e.g., given an 8-bit integer $x = -128_{10} = 1000\_0000_2$, then $x \gg 1 = -64_{10} = 1100\_0000_2$ and $x \ggg 1 = 64_{10} = 0100\_0000_2$. The function $sizeof(x)$ returns the bit-width of the operand $x$.

\begin{algorithm}
	\caption{bsort for signed integers}
	\SetKwInOut{Input}{Input}
	\SetKwInOut{Output}{Output}
	
	\Fn{bsortSigned(a, asc)}{
		\Input{$a$ - an array of signed integers\newline
			$asc$ - boolean indicating the sorting order
		}
		\Output{The same array $a$ but sorted}
		
		$m := 1 \ll (sizeof(a_0) - 1)$
		
		$k := singlePassBSort(a, m, 0, |a|, \neg asc)$
		
		$m' := m \ggg 1$
		
		$binaryQuickSort(a, m', 0, k, asc)$
		
		$binaryQuickSort(a, m', k, |a|, asc)$
		
		\KwRet{a}
	}
	\label{alg:bsort-signed}
\end{algorithm}


\textbf{Example for Algorithm 3:}

Let the input data be $a = [-8, -1, 3, -7, 2, 6, 3], asc = false$, then the tracing of $bsortSigned(a, asc)$ is shown in Figure~\ref{fig:bsort-signed-tracing}.

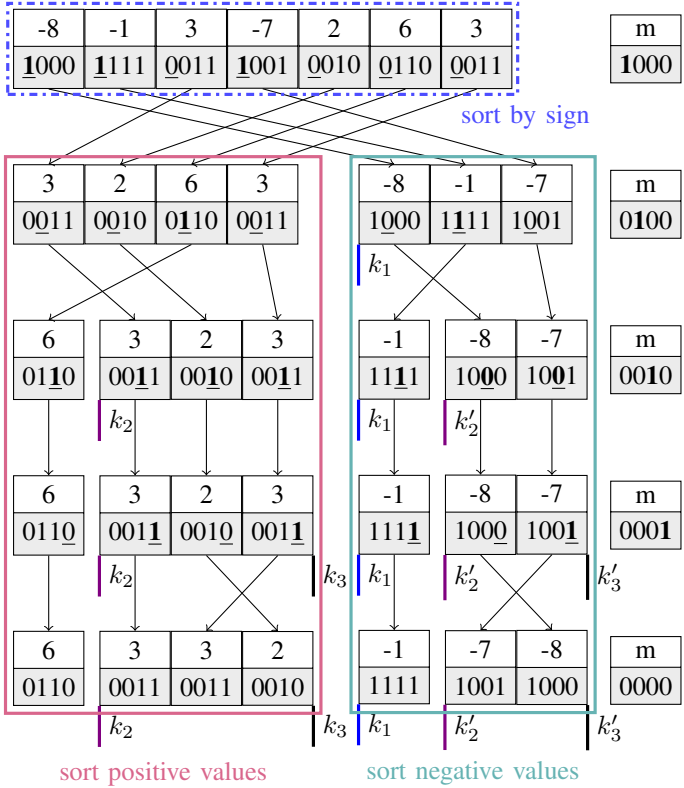
\begin{figure}[!t]
	\vspace{-5pt}
	\begin{tikzpicture}[
elem/.style={rectangle split, rectangle split parts=2, 
	draw, align=center, text width=7mm,
	rectangle split part fill={white,gray!15}}
]
\node[elem]                         (01--8)	    {-8 \nodepart{two} \under{\textbf{1}}000 };
\node[elem, right=0cm of 01--8]     (01--1)		{-1 \nodepart{two} \under{\textbf{1}}111 };
\node[elem, right=0cm of 01--1]     (01-3-1)	{3  \nodepart{two} \under{0}011 };
\node[elem, right=0cm of 01-3-1]    (01--7)		{-7 \nodepart{two} \under{\textbf{1}}001 };
\node[elem, right=0cm of 01--7]     (01-2)		{2  \nodepart{two} \under{0}010 };
\node[elem, right=0cm of 01-2]      (01-6)      {6  \nodepart{two} \under{0}110 };
\node[elem, right=0cm of 01-6]      (01-3-2)	{3  \nodepart{two} \under{0}011 };

\node[elem]	(01-m)	[right=1.3cm of 01-3-2]	{m \nodepart{two} \textbf{1}000 };

\node[elem, below=of 01--8]         (02-3-1)	{3  \nodepart{two} 0\under{0}11 };
\node[elem, right=0cm of 02-3-1]    (02-2)		{2  \nodepart{two} 0\under{0}10 };
\node[elem, right=0cm of 02-2]      (02-6)      {6  \nodepart{two} 0\under{\textbf{1}}10 };
\node[elem, right=0cm of 02-6]      (02-3-2)	{3  \nodepart{two} 0\under{0}11 };
\node[elem, right=0.8cm of 02-3-2]  (02--8)		{-8 \nodepart{two} 1\under{0}00 };
\node[elem, right=0cm of 02--8]     (02--1)	    {-1 \nodepart{two} 1\under{\textbf{1}}11 };
\node[elem, right=0cm of 02--1]     (02--7)		{-7 \nodepart{two} 1\under{0}01 };

\node (k1) [text width=7mm, below=0cm of 02--8] {$k_1$};
\draw[very thick, blue] (k1.south west) |- (k1.north west);

\node[elem]	(02-m)	[below=1.15cm of 01-m]	{m \nodepart{two} 0\textbf{1}00 };

\draw[->] (01-3-1.south) -- (02-3-1.north);
\draw[->] (01--7.south)  -- (02--7.north);
\draw[->] (01-6.south)   -- (02-6.north);
\draw[->] (01--8.south)  -- (02--8.north);
\draw[->] (01-2.south)   -- (02-2.north);
\draw[->] (01-3-2.south) -- (02-3-2.north);
\draw[->] (01--1.south)  -- (02--1.north);

\node[elem, below=of 02-3-1]        (03-6)      {6  \nodepart{two} 01\under{\textbf{1}}0 };
\node[elem, right=0.2cm of 03-6]    (03-3-1)	{3  \nodepart{two} 00\under{\textbf{1}}1 };
\node[elem, right=0cm of 03-3-1]    (03-2)		{2  \nodepart{two} 00\under{\textbf{1}}0 };
\node[elem, right=0cm of 03-2]      (03-3-2)	{3  \nodepart{two} 00\under{\textbf{1}}1 };
\node[elem, below=of 02--8]         (03--1)	    {-1 \nodepart{two} 11\under{\textbf{1}}1 };
\node[elem, right=0.2cm of 03--1]   (03--8)		{-8 \nodepart{two} 10\under{\textbf{0}}0 };
\node[elem, right=0cm of 03--8]     (03--7)		{-7 \nodepart{two} 10\under{\textbf{0}}1 };

\node (k1) [text width=7mm, below=0cm of 03--1] {$k_1$};
\draw[very thick, blue] (k1.south west) |- (k1.north west);

\node (k2) [text width=7mm, below=0cm of 03-3-1] {$k_2$};
\draw[very thick, violet] (k2.south west) |- (k2.north west);

\node (k2') [text width=7mm, below=0cm of 03--8] {$k_2'$};
\draw[very thick, violet] (k2'.south west) |- (k2'.north west);

\node[elem]	(03-m)	[below=1.15cm of 02-m]	{m \nodepart{two} 00\textbf{1}0 };

\draw[->] (02-3-1.south) -- (03-3-1.north);
\draw[->] (02--7.south)  -- (03--7.north);
\draw[->] (02-6.south)   -- (03-6.north);
\draw[->] (02--8.south)  -- (03--8.north);
\draw[->] (02-2.south)   -- (03-2.north);
\draw[->] (02-3-2.south) -- (03-3-2.north);
\draw[->] (02--1.south)  -- (03--1.north);

\node[elem, below=of 03-6]          (04-6)      {6  \nodepart{two} 011\under{0} };
\node[elem, right=0.2cm of 04-6]    (04-3-1)	{3  \nodepart{two} 001\under{\textbf{1}} };
\node[elem, right=0cm of 04-3-1]    (04-2)		{2  \nodepart{two} 001\under{0} };
\node[elem, right=0cm of 04-2]      (04-3-2)	{3  \nodepart{two} 001\under{\textbf{1}} };
\node[elem, below=of 03--1]         (04--1)	    {-1 \nodepart{two} 111\under{\textbf{1}} };
\node[elem, right=0.2cm of 04--1]   (04--8)		{-8 \nodepart{two} 100\under{0} };
\node[elem, right=0cm of 04--8]     (04--7)		{-7 \nodepart{two} 100\under{\textbf{1}} };

\node (k1) [text width=7mm, below=0cm of 04--1] {$k_1$};
\draw[very thick, blue] (k1.south west) |- (k1.north west);

\node (k2) [text width=7mm, below=0cm of 04-3-1] {$k_2$};
\draw[very thick, violet] (k2.south west) |- (k2.north west);

\node (k2') [text width=7mm, below=0cm of 04--8] {$k_2'$};
\draw[very thick, violet] (k2'.south west) |- (k2'.north west);

\node (k3) [text width=7mm, below right=0cm of 04-3-2] {$k_3$};
\draw[very thick, black] (k3.south west) |- (k3.north west);

\node (k3') [text width=7mm, below right=0cm of 04--7] {$k_3'$};
\draw[very thick, black] (k3'.south west) |- (k3'.north west);

\node[elem]	(04-m)	[below=1.15cm of 03-m]	{m \nodepart{two} 000\textbf{1} };

\draw[->] (03-3-1.south) -- (04-3-1.north);
\draw[->] (03--7.south)  -- (04--7.north);
\draw[->] (03-6.south)   -- (04-6.north);
\draw[->] (03--8.south)  -- (04--8.north);
\draw[->] (03-2.south)   -- (04-2.north);
\draw[->] (03-3-2.south) -- (04-3-2.north);
\draw[->] (03--1.south)  -- (04--1.north);

\node[elem, below=of 04-6]          (05-6)      {6  \nodepart{two} 0110 };
\node[elem, right=0.2cm of 05-6]    (05-3-1)	{3  \nodepart{two} 0011 };
\node[elem, right=0cm of 05-3-1]    (05-3-2)	{3  \nodepart{two} 0011 };
\node[elem, right=0cm of 05-3-2]    (05-2)		{2  \nodepart{two} 0010 };
\node[elem, below=of 04--1]         (05--1)	    {-1 \nodepart{two} 1111 };
\node[elem, below=of 04--8]         (05--7)		{-7 \nodepart{two} 1001 };
\node[elem, right=0cm of 05--7]     (05--8)		{-8 \nodepart{two} 1000 };

\node (k1) [text width=7mm, below=0cm of 05--1] {$k_1$};
\draw[very thick, blue] (k1.south west) |- (k1.north west);

\node (k2) [text width=7mm, below=0cm of 05-3-1] {$k_2$};
\draw[very thick, violet] (k2.south west) |- (k2.north west);

\node (k2') [text width=7mm, below=0cm of 05--7] {$k_2'$};
\draw[very thick, violet] (k2'.south west) |- (k2'.north west);

\node (k3) [text width=7mm, below right=0cm of 05-2] {$k_3$};
\draw[very thick, black] (k3.south west) |- (k3.north west);

\node (k3') [text width=7mm, below right=0cm of 05--8] {$k_3'$};
\draw[very thick, black] (k3'.south west) |- (k3'.north west);

\node[elem]	(05-m)	[below=1.15cm of 04-m]	{m \nodepart{two} 0000 };

\draw[->] (04-3-1.south) -- (05-3-1.north);
\draw[->] (04--7.south)  -- (05--7.north);
\draw[->] (04-6.south)   -- (05-6.north);
\draw[->] (04--8.south)  -- (05--8.north);
\draw[->] (04-2.south)   -- (05-2.north);
\draw[->] (04-3-2.south) -- (05-3-2.north);
\draw[->] (04--1.south)  -- (05--1.north);

\node[draw=blue!70, very thick, dash dot, inner sep=2pt, label={[label distance=0cm, color=blue!70]-14:sort by sign}, fit=(01--8) (01-3-2)] {};

\node[draw=purple!60, very thick, inner sep=0.1cm, label={[label distance=0.5cm, color=purple!60]-90:sort positive values}, fit=(02-3-1) (05-2)] {};

\node[draw=teal!60, very thick, inner sep=0.1cm, label={[label distance=0.5cm, color=teal!60]-90:sort negative values}, fit=(02--8) (05--8)] {};
	\end{tikzpicture}
	\caption{Tracing of bsort algorithm for signed integers using two's complement notation.}
	\label{fig:bsort-signed-tracing}
	\vspace{-10pt}
\end{figure}

\subsection{bsort for floating-point values}

\begin{theorem}
	\label{theorem:rational-to-integer}
	Let $x$ be a rational number with a finite fractional representation in an integer base $b > 1$. Then $x$ can be expressed in the form:
	\begin{equation}
		x = s \cdot m \cdot b^p
	\end{equation}
	where $s \in \{-1, 1\}$ is the sign, $m \in \mathbb{N}_0$ the mantissa, and $p \in \mathbb{Z}_0^{-}$ the exponent.
\end{theorem}

\begin{proof}
	Any such number $x$ can be decomposed into its sign $s \in \{-1, 1\}$, its integer part $a \in \mathbb{N}_0$, and its fractional part $f \in [0, 1)$. The fractional part $f$ consists of a finite sequence of $q$ digits.
	
	We can express $x$ as $x = s \cdot (a + f)$. Let the integer formed by the $q$ fractional digits be $c$. This is equivalent to $c = f \cdot b^q$, which implies $f = c/b^q$.
	
	Substituting this back into the expression for $x$:
	\[ x = s \cdot \left(a + \frac{c}{b^q}\right) = s \cdot \left(\frac{a \cdot b^q + c}{b^q}\right) = s \cdot (a \cdot b^q + c) \cdot b^{-q} \]
	Since $a, b^q, c \in \mathbb{N}_0$, then $(a \cdot b^q + c) \in \mathbb{N}_0$ by the closure properties of addition and multiplication. Let this integer be $m = a \cdot b^q + c$.
	
	Furthermore, let $p = -q$. Since $q \in \mathbb{N}_0$, its additive inverse $p \in \mathbb{Z}_0^{-}$. Thus, we can write:
	\[ x = s \cdot m \cdot b^p \]
	This completes the proof.
\end{proof}

Examples of Theorem~\ref{theorem:rational-to-integer}:

\begin{itemize}
    \item $112.9_{10} = 1 \cdot 1129_{10} \cdot 10 ^ {-1}$
    \item $-8.348975_{10} = -1 \cdot 8348975_{10} \cdot 10 ^ {-6}$
    \item $111.11_2 = 1 (2 ^ 2 \cdot 111_2 + 11_2) 2 ^ {-2} = 11111_2 \cdot 2 ^ {-2}$
    \item $-101.1_2 = -1 (2 ^ 1 \cdot 101_2 + 1_2) 2 ^ {-1} = -1 \cdot 1011_2 \cdot 2 ^ {-1}$
    \item $4789_{10} = 1 \cdot 4789_{10} \cdot 10 ^ 0$
\end{itemize}

Sorting floating-point values, represented as described in Theorem~\ref{theorem:rational-to-integer}, is accomplished through a sequential, multi-pass sort on the components $s$, $p$, and $m$, in that order.

\begin{enumerate}
	\item Sort by sign, $s$: The first pass sorts the entire array on the sign component, which partitions the data into two primary sections: one for all negative numbers and one for all non-negative numbers.
	\item Sort by exponent, $p$: Next, each of those two sections is sorted independently by the exponent $p$. This creates finer-grained sub-partitions where all elements share the same sign and exponent. To maintain correct numerical order, the sort direction for the exponents of negative numbers must be the inverse of the direction for positive numbers. For example, to sort the final array in ascending order, the exponents of negative numbers must be sorted in descending order.
	\item Sort by mantissa, $m$: Since all elements within a sub-partition already have the same sign and exponent, ordering them by their mantissa correctly places them in their final sorted position as it is the same as sorting unsigned integers. The sort direction also depends on the sign.
\end{enumerate}

This three-pass procedure correctly sorts the entire array by progressively refining the order of the elements. Formal proofs are provided in Theorems~\ref{theorem:order-necessity} and~\ref{theorem:correctness-float}.

\begin{theorem}[Necessity of sort order]
	\label{theorem:order-necessity}
	The correctness of the \textit{bsortFloat} algorithm is dependent on the hierarchical sorting order of sign, then exponent, then mantissa.
\end{theorem}

\begin{proof}
	We will prove that any hierarchical sorting order other than sign $\rightarrow$ exponent $\rightarrow$ mantissa will fail to produce a correctly sorted array for all inputs. The numerical value of a number is determined first by its sign, then by the magnitude dictated by its exponent, and finally refined by its mantissa. Let's prove this by contradiction by analyzing the two possible ways the required order can be violated.
	
	\textbf{Case 1: The sign component ($s$) is not the primary sort key.}
	Assume the primary sort key is either the exponent ($p$) or the mantissa ($m$). Consider the input array $[x, y, z]~=~[-2.0, 0.0, 0.5]$. Their representations according to Theorem~\ref{theorem:rational-to-integer} are $x = -1 \cdot 10_2 \cdot 2^0$, $y = 1 \cdot 0_2 \cdot 2^0$ and $z = 1 \cdot 1_2 \cdot 2^{-1}$. Suppose we aim to sort the array in ascending order.
	
	If the algorithm sorts by mantissa first, it will incorrectly imply that $y < z < x$ because $m_y < m_z < m_x$ ($0_2~<~1_2~<~10_2$). Because the initial partitioning by mantissa places each element into a singleton partition, the subsequent recursive sorts on each one, will not change their order. The resulting array, $[0.0, 0.5, -2.0]$, fails to satisfy the ascending sort invariant.
	
	If the algorithm sorts by exponent first, it will incorrectly imply that $z < x$ because $p_z < p_x$ ($-1~<~0$). The initial partitioning by exponent yields two partitions, in the following relative order: $[z]$ and $[x, y]$. Since $z$ is isolated in its own partition, no subsequent recursion on mantissa or sign can correct its position relative to $x$. The final output, $[0.5, 0.0, -2.0]$, is also incorrectly sorted.
	
	These cases demonstrate that neither the exponent nor the mantissa maintains global monotonicity across the set of signed numbers. Consequently, the sign bit must serve as the primary sort key; this ensures that the most fundamental numerical boundary (the distinction between negative and non-negative values) is established before the magnitudes of exponents and mantissas are evaluated.

	\textbf{Case 2: The sign is the primary key, but the exponent ($p$) is not the secondary key.}
	This implies the sort order must be sign $\rightarrow$ mantissa $\rightarrow$ exponent. To test this, we only need to consider numbers of the same sign. Let the input array be $[x, y] = [0.75, 2.0]$. Again, suppose we aim to sort the array in ascending order. Their representations are $x = 1 \cdot 11_2 \cdot 2^{-2}$ and $y = 1 \cdot 10_2 \cdot 2^0$.
	
	The algorithm first correctly places both into a positive numbers partition. It then proceeds to sort such partition by mantissa. Since $m_y < m_x$ ($10_2 < 11_2$), the algorithm places $y$ before $x$, resulting in the order $[2.0, 0.75]$. As these elements are now separated in their respective partition, the final sort by exponent will not change their order. The result violates the ascending sort invariant.
	
	The exponent determines a number's order of magnitude, while the mantissa refines its value within that magnitude. An algorithm must therefore evaluate the exponent before the mantissa.
\end{proof}

\begin{theorem}[Sufficiency of sort order]
	\label{theorem:correctness-float}
	The \textit{bsortFloat} algorithm, in conjunction with its sub-procedures, correctly sorts an array of floating-point numbers that adhere to the representation defined in Theorem~\ref{theorem:rational-to-integer}.
\end{theorem}

\begin{proof}
	The proof is structured in three parts, mirroring the algorithm's execution flow: (1) the initial sign partition, (2) the recursive sort on the exponent field, and (3) the final sort on the mantissa field.
	
	\textbf{1. Correctness of Sign sort:}
	The logic for partitioning floating-point numbers by sign is identical to the method established for signed integers (Section~\ref{subsection:bsort for signed integers}). As demonstrated in that section, the bitwise order of the sign bit is the inverse of the required numerical order for an ascending sort. A naive ascending pass would incorrectly place non-negative numbers before negative ones.
	
	Therefore, $bsortFloat$ employs the same solution: it inverts the sort direction for the initial pass by calling \textit{singlePassBSort} with \textit{$\neg$asc}. This leverages the exact same mechanism to correctly place all negative numbers before all non-negative numbers, establishing the correct primary order.
	
	\textbf{2. Correctness of Exponent sort:}
	We prove by induction that the recursive calls in \textit{bsortF} correctly sort each same-signed partition by its exponent. Let $w_p$ be the number of bits in the exponent. Let $P_E(k)$ be the inductive hypothesis for the exponent sort:
	
	\textbf{Inductive Hypothesis $P_E(k)$:} After processing the $k$ most significant bits of the exponent (for $1 \le k \le w_p$), the partition is correctly sorted according to the numerical value represented by those $k$ bits.
	
	\textbf{Base Case ($k=1$):} Let the algorithm process only the most significant bit ($b_k$). By definition $singlePassBSort$ partitions the array into two sets: $S_0 = \{x | b_k = 0\}$ and $S_1 = \{x | b_k = 1\}$.
	If the sorting direction is ascending, $S_0$ is placed before $S_1$, and vice versa if descending. In both cases, the relative order of $S_0$ and $S_1$ satisfies the numerical requirement for that single bit. Thus, the array is correctly partitioned with respect to the $k=1$ bit, and $P_E(1)$ holds
	
	\textbf{Inductive Step:} Assume the inductive hypothesis $P_E(k)$ is true for some $k < w_p$. This means the partition is correctly sorted based on the lexicographical value of the first $k$ bits of the exponent. All elements that share an identical $k$-bit prefix are now in the same sub-partition.
	
	The algorithm then processes the $(k+1)$-th bit. Consider any one of the sub-partitions created in the previous step. \textit{singlePassBSort} is called on this sub-partition, dividing it into two smaller groups: a $0$-group where the $(k+1)$-th bit is $0$, and a $1$-group where it is $1$.
	
	The procedure places the $0$-group before the $1$-group (for a standard ascending sort). Let $x$ and $y$ be any element from the $0$-group and $1$-group, respectively. The $(k+1)$-bit prefix of $x$ is now lexicographically smaller than the prefix of $y$. By placing all such $x$ before all such $y$, the algorithm correctly extends the sorted property to the $(k+1)$-th bit.
	
	This demonstrates that if the array is correctly sorted by $k$ bits, the algorithm will correctly sort it by $k+1$ bits. Thus, if $P_E(k)$ is true, then $P_E(k+1)$ is also true.
	
	\textbf{3. Correctness of Mantissa sort:}
	After the exponent sort, \textit{bsortF} calls \textit{binaryQuickSort} on each sub-partition. Within any such sub-partition, all elements have an identical sign and exponent, so their numerical order depends only on their mantissa. Sorting by mantissa is equivalent to sorting unsigned integers, for which \textit{binaryQuickSort} is correct. The sort direction is chosen based on the sign of the partition, ensuring the final order is correct.
	
	\textbf{Conclusion:}
	Since the array is (1) correctly partitioned by sign, (2) correctly sub-partitioned by exponent, and (3) each resulting sub-partition is correctly sorted by mantissa, the \textit{bsortFloat} algorithm is correct.
\end{proof}

Algorithms~4 and~5 illustrate this sorting procedure. It is important to clarify that in these algorithms $sizeof\_exponent(x)$, $sizeof\_mantissa(x)$ and $sizeof(x)$, return the bit-width of the exponent, mantissa and the actual value of $x$, respectively. Consequently, the relationship between these widths is defined as $size\_of(x) = sizeof\_mantissa(x) + sizeof\_exponent(x) +~1$, where the constant term accounts for the sign bit.

It is assumed that the sign, exponent and mantissa are stored in the order in which they were mentioned (so most significant bit corresponds to the sign, and least significant bit is part of the mantissa). However, the algorithm, with some modifications, still works if the arrangement is different.

\begin{algorithm}
    \caption{bsort for floating-point values}
    \SetKwInOut{Input}{Input}
    \SetKwInOut{Output}{Output}
    
    \Fn{bsortFloat(a, asc)}{
        \Input{$a$ - an array of floating-point values\newline
            $asc$ - boolean indicating the sorting order
        }
        \Output{The same array $a$ but sorted}
        
        \tcp{sort by sign}
        $m := 1 \ll (sizeof(a_0) - 1)$
        
        $k := singlePassBSort(a, m, 0, |a|, \neg asc)$
        
        $ $ 
        
        \tcp{sort by exponent and mantissa}
        $m' := m \ggg 1$
        
        $bsortF(a, m', 0, k, asc)$
        
        $bsortF(a, m', k, |a|, asc)$
        
        \KwRet{a}
    }
    \label{alg:bsort-float}
\end{algorithm}

\begin{algorithm}[ht!]
    \caption{bsort for floating-point values, ignoring sign}
    \SetKwInOut{Input}{Input}
    \SetKwInOut{Output}{Output}
    
    \Fn{bsortF(a, m, ps, pe, asc)}{
        \Input{$a$ - The array of numbers to be sorted; $a_i \in a\ \forall\ i \in \mathbb{Z}: 0 \leq i < |a|$ \newline
            $m$ - the current mask; $\exists k \in \mathbb{N}_0: m = 2^k $\newline
            $ps, pe$ - the indices at which the partition to be sorted starts (inclusive) and ends (exclusive); $ps, pe \in \mathbb{Z}: 0 \leq ps \leq pe < |a|$\newline
            $asc$ - boolean indicating if the order is ascending or not
        }
        \Output{An index $k \in [ps, pe]$ such that $a_i \wedge m = 0, a_j \wedge m = m\ \forall i \in [ps, k), j \in [k, pe]$ if and only if $asc$ is true, or $a_i \wedge m = m, a_j \wedge m = 0$ otherwise}
        
        \tcp{sort by exponent}
        $k := singlePassBSort(a, m, ps, pe, asc)$
        
        $ $ 
        
        \tcp{if the current mask is the last bit of the exponent, then, now is time to sort only by the mantissa}
        \If{$m = (1 \ll sizeof\_mantissa(a_0))$}{
            $m' := m \ggg 1$
            
            $binaryQuickSort(a, m', ps, k, asc)$
            
            $binaryQuickSort(a, m', k, pe, asc)$
            
            \KwRet{a}
        }
        
        $m' := m \ggg 1$
        
        $bsortF(a, m', ps, k, asc)$
        
        $bsortF(a, m', k, pe, asc)$
        
        \KwRet{a}
    }
    \label{alg:bsort-f}
\end{algorithm}

\subsection{bsort for IEEE-754-represented floating-point values}

IEEE-754 standard represents floating-point values in a similar manner to Theorem~\ref{theorem:rational-to-integer}. For 32-bit floating-point values, the sign bit $s$ is stored in the left-most bit, the exponent $p$ is stored in the next 8 bits, and lastly the mantissa $m$ is stored in the 23 bits remaining. A similar arrangement is used for 64-bit, 128-bit, or any other word-size floating-point values~\cite{ieee-754}. The general arrangement is illustrated in Figure~\ref{fig:ieee754}

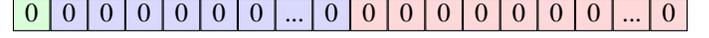
\begin{figure}[H]
    \begin{tikzpicture}
[
    sign/.style={rectangle, draw, align=center, text width=2.5mm, text height=2mm, fill=green!15},
    exponent/.style={rectangle, draw, align=center, text width=2.5mm, text height=2mm, fill=blue!15},
    mantissa/.style={rectangle, draw, align=center, text width=2.5mm, text height=2mm, fill=red!15}
]

\node[sign]                                 (float-sign)    {0};
\node[exponent, right=0cm of float-sign]    (float-exp-1)   {0};
\node[exponent, right=0cm of float-exp-1]   (float-exp-2)   {0};
\node[exponent, right=0cm of float-exp-2]   (float-exp-3)   {0};
\node[exponent, right=0cm of float-exp-3]   (float-exp-4)   {0};
\node[exponent, right=0cm of float-exp-4]   (float-exp-5)   {0};
\node[exponent, right=0cm of float-exp-5]   (float-exp-6)   {0};
\node[exponent, right=0cm of float-exp-6]   (float-exp-7)   {...};
\node[exponent, right=0cm of float-exp-7]   (float-exp-8)   {0};
\node[mantissa, right=0cm of float-exp-8]   (float-man-1)   {0};
\node[mantissa, right=0cm of float-man-1]   (float-man-2)   {0};
\node[mantissa, right=0cm of float-man-2]   (float-man-3)   {0};
\node[mantissa, right=0cm of float-man-3]   (float-man-4)   {0};
\node[mantissa, right=0cm of float-man-4]   (float-man-5)   {0};
\node[mantissa, right=0cm of float-man-5]   (float-man-6)   {0};
\node[mantissa, right=0cm of float-man-6]   (float-man-7)   {0};
\node[mantissa, right=0cm of float-man-7]   (float-man-8)   {...};
\node[mantissa, right=0cm of float-man-8]   (float-man-d)   {0};
    \end{tikzpicture}
    \caption{IEEE-754 floating-point values representation. In green the sign bit, in blue the exponent bits, and in red the mantissa bits.}
    \label{fig:ieee754}
\end{figure}

There are however some major differences between Theorem~\ref{theorem:rational-to-integer} and the IEEE-754 representation because the former has specific representations for $+\infty$, $-\infty$, \texttt{NaN}, $+0$, $-0$, and it also adds a bias equal to $2^{k - 1} - 1$ to the exponent, where $k$ is the number of bits in the exponent, and the mantissa includes an implicit leading bit of 1 for normalized values, which does not affect the lexicographical ordering of the stored bits~\cite{ieee-754}. But, these differences do not affect the correctness of Algorithms~4 and~5, and the proof for that follows directly from the fact that adding a constant offset is a monotonic transformation, which inherently preserves the numerical ordering.

Regarding the special values:

\begin{itemize}
    \item $+\infty$ is represented with the sign bit not set, all the exponent bits set, and none of the mantissa bits set. So, it will be put inside the partition of non-negative numbers with all the exponent bits set, which are either $+\infty$ or \texttt{NaN}, so the position in which $+\infty$ is put in the sorted array will imply that for any element $a_i \notin \{+\infty, NaN\}$ in the array $a$, $a_i < +\infty$ holds.
    \item $-\infty$ is represented with the sign bit set, all the exponent bits set, and none of the mantissa bits set. So, it will be put inside the partition of negative numbers with all the exponent bits set, which are either $-\infty$ or \texttt{NaN}, so the position in which $-\infty$ is put in the sorted array will imply that for any element $a_i \notin \{-\infty, NaN\}$ in the array $a$, $a_i > -\infty$ holds.
    \item \texttt{NaN} is represented with the sign bit set or not set, all the exponent bits set, and at least one of the mantissa bits set. So, it will be put at the extremities of the array, with their relative order determined by the mantissa payload.
    \item $+0$ and $-0$ are represented with the sign bit not set and set, respectively, and with none of the mantissa or exponent bits set. So, $-0$ will be put inside the negative numbers partition, and because there is no exponent or mantissa bit set it will be put at the end of the partition when the order is ascending, and at the beginning when the order is descending. Thus, implying that for any negative element $a_i$, $a_i < -0$ is true.
	Similarly, $a_i > +0$ holds for any positive element $a_i$. This bitwise approach naturally enforces the order $-0 < +0$.
\end{itemize}
\pagebreak

\textbf{Example:}

Let $asc=true$, $a = [1.75, 1.25, -2.5, -\infty]$ and, for sake of brevity and clarity, let's assume that 6 bits are used to store a single floating-point value: 1 for the sign bit, 3 for the exponent (so bias is 3), and 2 for the mantissa; then $a = [001111_2, 001101_2, 110001_2, 111100_2]$, and the tracing of $bsortFloat(a, asc)$ is shown in Figure~\ref{fig:bsort-floating-point-tracing}.

\begin{figure}[!t]
    \begin{tikzpicture}[
elem/.style={rectangle split, rectangle split parts=2, 
    draw, align=center, text width=11mm,
    rectangle split part fill={white,gray!15}}
]
\node[elem]                       (01-a)	{1.75       \nodepart{two} \under{0}01111};
\node[elem, right=0cm of 01-a]    (01-b)	{1.25       \nodepart{two} \under{0}01101};
\node[elem, right=0cm of 01-b]    (01-c)	{-2.5       \nodepart{two} \under{\textbf{1}}10001};
\node[elem, right=0cm of 01-c]    (01-d)	{$-\infty$  \nodepart{two} \under{\textbf{1}}11100 };

\node[elem]	(01-m)	[right=1.5cm of 01-d]	{m \nodepart{two} \textbf{1}00000 };

\node[elem, below=of 01-a]        (02-c)	{-2.5       \nodepart{two} 1\under{\textbf{1}}0001 };
\node[elem, right=0cm of 02-c]    (02-d)	{$-\infty$  \nodepart{two} 1\under{\textbf{1}}1100 };
\node[elem, right=0.75cm of 02-d] (02-a)	{1.75       \nodepart{two} 0\under{0}1111 };
\node[elem, right=0cm of 02-a]    (02-b)	{1.25       \nodepart{two} 0\under{0}1101 };

\node (k1) [text width=11mm, below=0cm of 02-a] {$k_1$};
\draw[very thick, blue] (k1.south west) |- (k1.north west);

\node[elem]	(02-m)	[below=1cm of 01-m]	{m \nodepart{two} 0\textbf{1}0000 };

\draw[->] (01-a.south) -- (02-a.north);
\draw[->] (01-b.south) -- (02-b.north);
\draw[->] (01-c.south) -- (02-c.north);
\draw[->] (01-d.south) -- (02-d.north);

\node[elem, below=of 02-c]        (03-c)	{-2.5       \nodepart{two} 11\under{0}001 };
\node[elem, right=0cm of 03-c]    (03-d)	{$-\infty$  \nodepart{two} 11\under{\textbf{1}}100 };
\node[elem, right=0.75cm of 03-d] (03-a)	{1.75       \nodepart{two} 00\under{\textbf{1}}111 };
\node[elem, right=0cm of 03-a]    (03-b)	{1.25       \nodepart{two} 00\under{\textbf{1}}101 };

\node (k1) [text width=11mm, below=0cm of 03-a] {$k_1$};
\draw[very thick, blue] (k1.south west) |- (k1.north west);

\node[elem]	(03-m)	[below=1.15cm of 02-m]	{m \nodepart{two} 00\textbf{1}000 };

\draw[->] (02-a.south) -- (03-a.north);
\draw[->] (02-b.south) -- (03-b.north);
\draw[->] (02-c.south) -- (03-c.north);
\draw[->] (02-d.south) -- (03-d.north);

\node[elem, below=of 03-c]          (04-d)	{$-\infty$  \nodepart{two} 111\under{\textbf{1}}00 };
\node[elem, right=0.25cm of 04-d]   (04-c)	{-2.5       \nodepart{two} 110\under{0}01 };
\node[elem, right=0.5cm of 04-c]    (04-a)	{1.75       \nodepart{two} 001\under{\textbf{1}}11 };
\node[elem, right=0cm of 04-a]      (04-b)	{1.25       \nodepart{two} 001\under{\textbf{1}}01 };

\node (k1) [text width=11mm, below=0cm of 04-a] {$k_1$};
\draw[very thick, blue] (k1.south west) |- (k1.north west);

\node (k3) [text width=11mm, below=0cm of 04-c] {$k_3$};
\draw[very thick, violet] (k3.south west) |- (k3.north west);

\node[elem]	(04-m)	[below=1.15cm of 03-m]	{m \nodepart{two} 000\textbf{1}00 };

\draw[->] (03-a.south) -- (04-a.north);
\draw[->] (03-b.south) -- (04-b.north);
\draw[->] (03-c.south) -- (04-c.north);
\draw[->] (03-d.south) -- (04-d.north);

\node[elem, below=of 04-d]          (05-d)	{$-\infty$  \nodepart{two} 1111\under{0}0 };
\node[elem, right=0.25cm of 05-d]   (05-c)	{-2.5       \nodepart{two} 1100\under{0}1 };
\node[elem, right=0.5cm of 05-c]    (05-a)	{1.75       \nodepart{two} 0011\under{\textbf{1}}1 };
\node[elem, right=0cm of 05-a]      (05-b)	{1.25       \nodepart{two} 0011\under{0}1 };

\node (k1) [text width=11mm, below=0cm of 05-a] {$k_1$};
\draw[very thick, blue] (k1.south west) |- (k1.north west);

\node (k3) [text width=11mm, below=0cm of 05-c] {$k_3$};
\draw[very thick, violet] (k3.south west) |- (k3.north west);

\node[elem]	(05-m)	[below=1.15cm of 04-m]	{m \nodepart{two} 0000\textbf{1}0 };

\draw[->] (04-a.south) -- (05-a.north);
\draw[->] (04-b.south) -- (05-b.north);
\draw[->] (04-c.south) -- (05-c.north);
\draw[->] (04-d.south) -- (05-d.north);

\node[elem, below=of 05-d]          (06-d)	{$-\infty$  \nodepart{two} 11110\under{0} };
\node[elem, right=0.25cm of 06-d]   (06-c)	{-2.5       \nodepart{two} 11000\under{\textbf{1}} };
\node[elem, right=0.5cm of 06-c]    (06-b)	{1.25       \nodepart{two} 00110\under{1}};
\node[elem, right=0.25cm of 06-b]   (06-a)	{1.75       \nodepart{two} 00111\under{\textbf{1}}};

\node (k1) [text width=11mm, below=0cm of 06-b] {$k_1$};
\draw[very thick, blue] (k1.south west) |- (k1.north west);

\node (k3) [text width=11mm, below=0cm of 06-c] {$k_3$};
\draw[very thick, violet] (k3.south west) |- (k3.north west);

\node (k5) [text width=11mm, below=0cm of 06-a] {$k_5$};
\draw[very thick, violet] (k5.south west) |- (k5.north west);

\node[elem]	(06-m)	[below=1.15cm of 05-m]	{m \nodepart{two} 00000\textbf{1} };

\draw[->] (05-a.south) -- (06-a.north);
\draw[->] (05-b.south) -- (06-b.north);
\draw[->] (05-c.south) -- (06-c.north);
\draw[->] (05-d.south) -- (06-d.north);

\node[draw=blue!70, very thick, dash dot, inner sep=2pt, label={[label distance=0cm, color=blue!70]-14:sort by sign}, fit=(01-a) (01-d)] {};

%

\node[draw=purple!60, very thick, dash dot, inner sep=0.1cm, label={[label distance=0.1cm, color=purple!60]-45:sort by exponent}, fit=(02-c) (04-b)] {};

\node[draw=teal!60, very thick, dash dot, inner sep=0.1cm, label={[label distance=0.1cm, color=teal!60]-30:sort by mantissa}, fit=(05-d) (06-a)] {};
    \end{tikzpicture}
    \caption{Tracing of bsort for floating-point values. For clarity, only relevant partition indexes ($k_1$, $k_3$, $k_5$) were illustrated.}
    \label{fig:bsort-floating-point-tracing}
\end{figure}
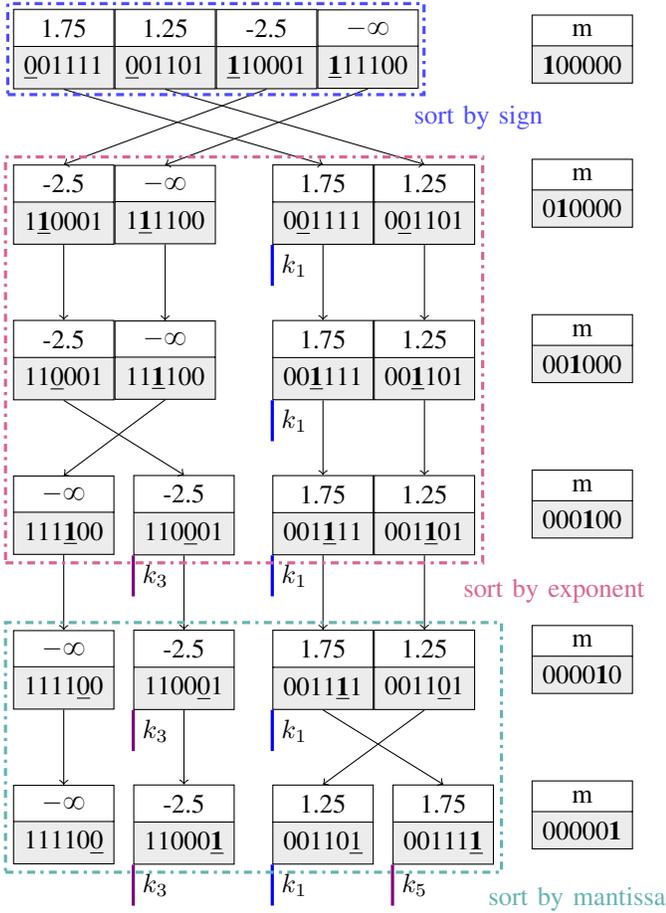

\section{Analysis}

\subsection{Theoretical analysis}

\begin{theorem}
	\label{theorem:bsort-runtime}
	The bsort algorithm sorts an array of $n$ elements with a word size $w$ in $O(nw)$ time.
\end{theorem}

\begin{proof}
	The algorithm's execution time is governed by the bitwise partitioning of the input.
	\begin{itemize}
		\item Per-bit pass: For each bit $b \in \{0, ..., w - 1\}$, the algorithm performs a single partitioning pass (e.g., $singlePassBSort$). This pass involves a linear scan of the current partition. Each element in the partition is inspected exactly once to check the value of its $b$-th bit, resulting in $O(n_p)$ operations, where $n_p$ is the number of elements in the current sub-partition.
		\item Total work per level: At any depth $d$ of the recursion tree (corresponding to a specific bit), the sum of all elements across all partitions is exactly $n$. Therefore, the total work performed at each level of the recursion is $O(n)$.
		\item Total depth: Since the algorithm recurses on the bit-width of the data, the recursion tree has a maximum depth of $w$.
	\end{itemize}
	
	Summing the work across all $w$ levels:
	\begin{equation}
		T(n) = \sum_{b = 1}^{w} O(n) = O (nw)
	\end{equation}
	
	As $w$ is constant for a given data type (e.g., $w=64$ for doubles), the runtime asymptotic behavior is linear with respect to the number of elements $n$, though formally expressed as $O(nw)$ to account for bit-depth.
\end{proof}

\begin{theorem}
	\label{theorem:bsort-space}
	The bsort algorithm operates with $O(w)$ auxiliary space
\end{theorem}

\begin{proof}
	This is established by two factors:
	
	\begin{itemize}
		\item In-place partitioning: The partitioning step (e.g., $singlePassBSort$) utilizes a two-pointer approach to swap elements within the original input array. This requires only a constant amount of auxiliary space, O(1), for temporary swap variables and pointers.
		\item Recursion depth: The algorithm performs a bitwise decomposition. For a word size of $w$ bits, the recursion tree has a maximum depth of $w$. Each stack frame stores a constant number of variables (pointers and bit masks). Therefore, the stack space consumed by recursive calls is $O(w)$.
	\end{itemize}
	
	As the algorithm does not allocate additional data structures proportional to the input size $n$, and the stack depth is bounded by $w$, the total auxiliary space complexity is $O(w)$. So, we can consider $bsort$ to be an in-place with respect to $n$ algorithm. 
\end{proof}

\subsection{Empirical analysis}

Introsort was selected as baseline for comparison given its status as the standard implementation in the C++ STL (std::sort). To evaluate \texttt{bsort} against modern state-of-the-art techniques, spreadsort and radix sort (ska\_sort) were included, both of which share a non-comparison-based architecture. While the standard C \texttt{qsort} implementation was initially considered, it was excluded from the final results due to consistently underperforming all other evaluated algorithms. This disparity arises primarily because \texttt{qsort} utilizes function pointers for comparisons, which inhibits the compiler from performing effective inlining and optimization, even at high optimization levels (\texttt{-O3})~\cite{effective-stl}.

The case of \texttt{qsort} serves as a salient example of the potential divergence between theoretical asymptotic behavior and practical performance; while its asymptotic bounds are well-established, architectural overheads and compiler limitations can significantly hinder its real-world efficiency.

A comparison of algorithms' properties can be found in Table~\ref{table:algorithm-comparison}.

\begin{table}[ht]
	\centering
	\caption{Comparative analysis of sorting algorithms.}
	\label{table:algorithm-comparison}
	\renewcommand{\arraystretch}{1.2}
	\small
	\begin{tabular}{|c|c|c|c|c|}
		\hline
		\textbf{Algorithm} & \textbf{Runtime } & \textbf{Memory$^+$ } & \textbf{In-place} & \textbf{Comp.}\\
		\hline
		Introsort$^*$ 		 & $O(n \log(n))$ 					& $O(\log(n))$ 	  & Yes  & Yes\\
		Spreadsort$^*$	  & $O(n (\frac{k}{s} + s))$	& $O(n)$  & No    & No\\
		ska\_sort$^*$		& $O(nw)$							  & $O(d \cdot k)$      	  & Yes   & No\\
		Quicksort				& $O(n \log(n))$				   & $O(\log(n))$	& Yes   & Yes\\
		bsort						& $O(nw)$							 & $O(w)$			 & Yes  & No\\
		\hline
	\end{tabular}
	\vspace{5pt}
	\begin{flushleft}
		The asterisk ($^*$) denotes hybrid algorithms. Consequently, their runtime, memory usage, and comparison-based behaviors vary depending on the specific subroutine executed.
	\end{flushleft}
\end{table}

$bsort$ as described in this paper, i.e. with no major modifications or optimizations, was implemented in C++ and was compared with the STL implementation of introsort (\texttt{std::sort}), Boost library implementation of spread sort (\texttt{boost::sort::spreadsort}), and an optimized implementation of radix sort (\texttt{ska\_sort}) in a 64-bit GNU/Linux environment with 48 GB of RAM and an Intel i5-8350U processor. The list of word size depending on the data type can be found in Table~\ref{table:word-size}. To mitigate the impact of external factors like transient system load, and to ensure statistical significance and consistency, GoogleBenchmark v1.8.0 was used. Compilation was done using gcc 12.3.1 with optimization flag \texttt{-O3}.

\begin{table}[htbp]
    \caption{Word size per data type user for testing}
    \begin{center}
        \begin{tabular}{|c|c|}
            \hline
            \textbf{Data type} & \textbf{Word size}\\
            \hline
            \texttt{char}        & 8\\
            \texttt{short}       & 16\\
            \texttt{int}         & 32\\
            \texttt{long long}   & 64\\
            \texttt{float}       & 32\\
            \texttt{double}      & 64\\
            \hline
        \end{tabular}
        \label{table:word-size}
    \end{center}
\end{table}

The analysis was performed on the data collected from various tests where each test case consisted of sorting an array of a certain size in $\{10^5, 10^6, 10^7, 10^8, 10^9, 5\cdot10^{9}\}$ and a certain data type in \{\texttt{char}, \texttt{short}, \texttt{int}, \texttt{long long}, \texttt{float}, \texttt{double}\}, using all four algorithms previously mentioned. The source code, including implementation, testing, and the \LaTeX sources for this manuscript, are publicly available at \url{https://benjaminguzman.dev}. Results comparing data type and running time are found in Figures~\ref{fig:comparison-by-ints} and~\ref{fig:comparison-by-floats}.

\begin{figure}[ht!]
	\centering
	\newcommand{\plotwidth}{0.95\columnwidth}
	\begin{subfigure}{\plotwidth}
		\includegraphics[width=\linewidth]{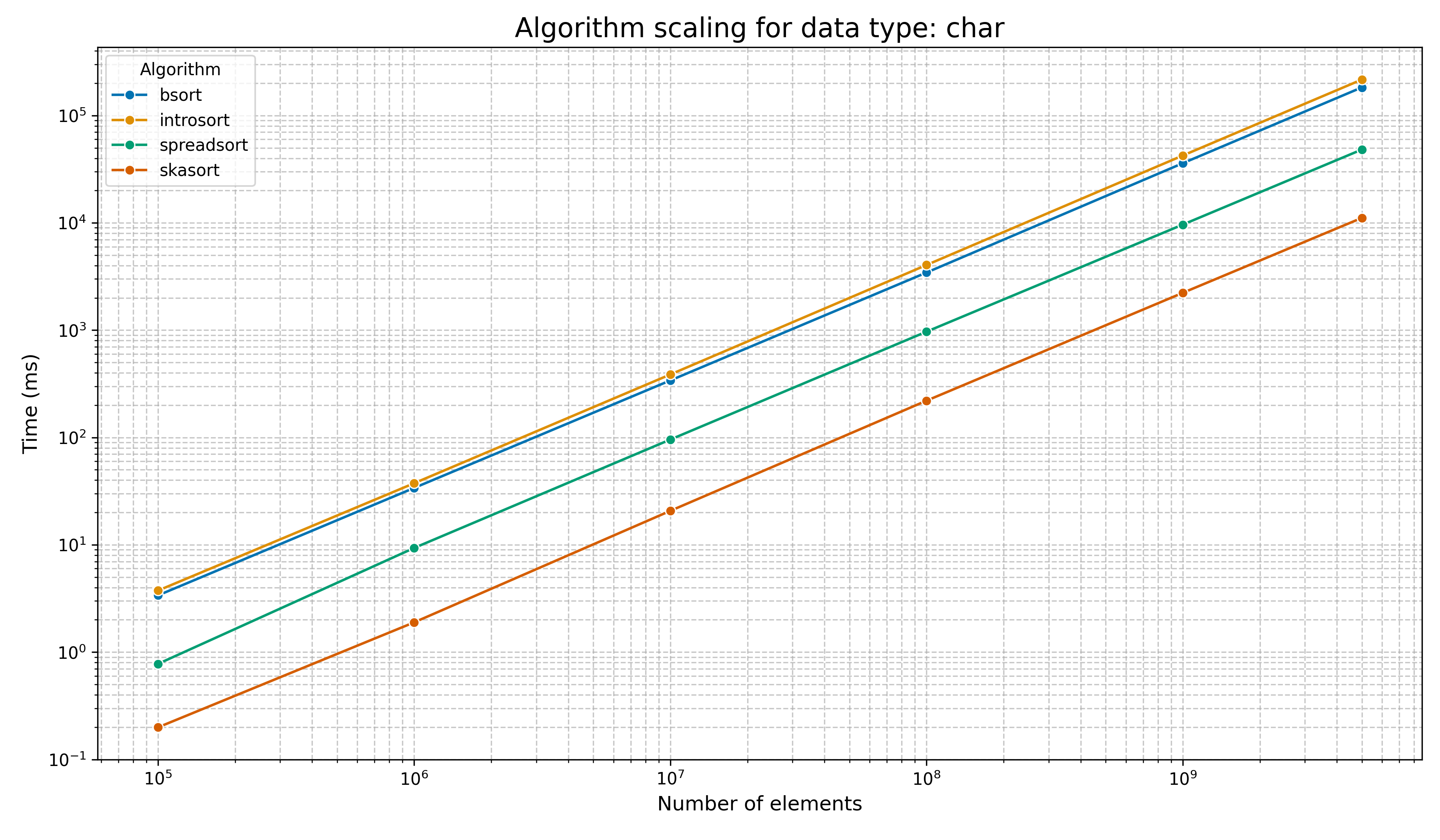}
	\end{subfigure}
	\begin{subfigure}{\plotwidth}
		\includegraphics[width=\linewidth]{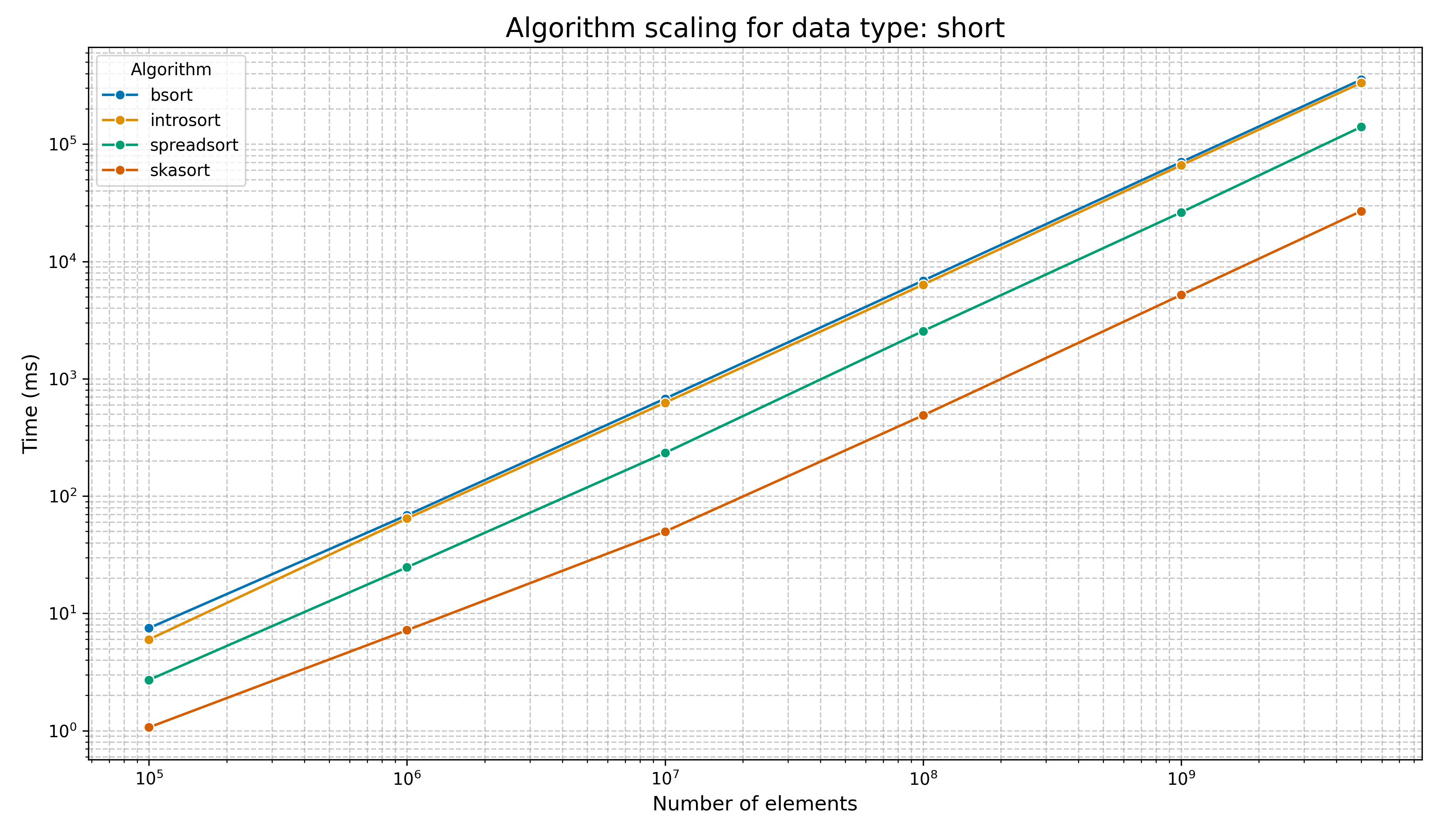}
	\end{subfigure}
	\begin{subfigure}{\plotwidth}
		\includegraphics[width=\linewidth]{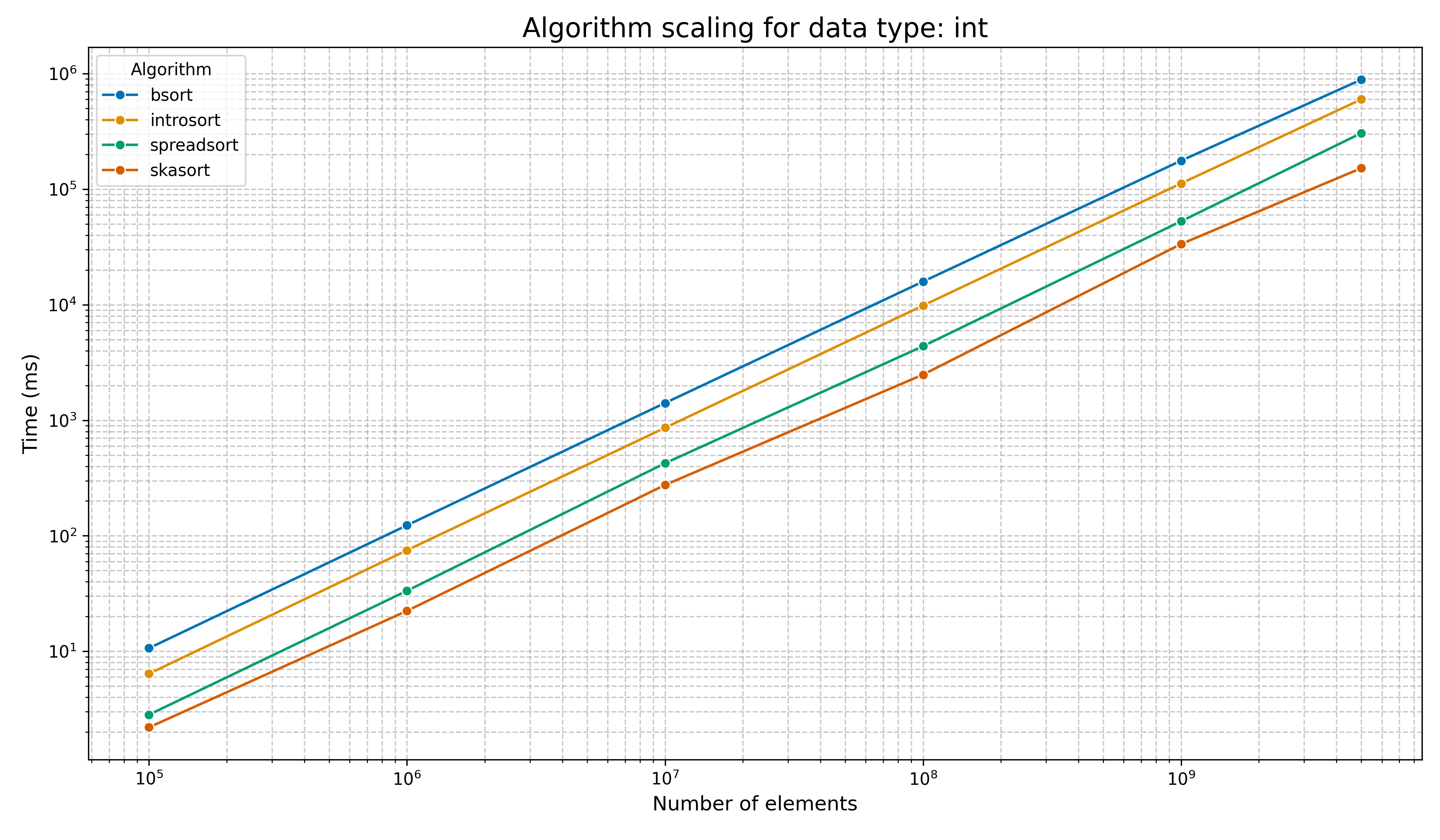}
	\end{subfigure}
	\begin{subfigure}{\plotwidth}
		\includegraphics[width=\linewidth]{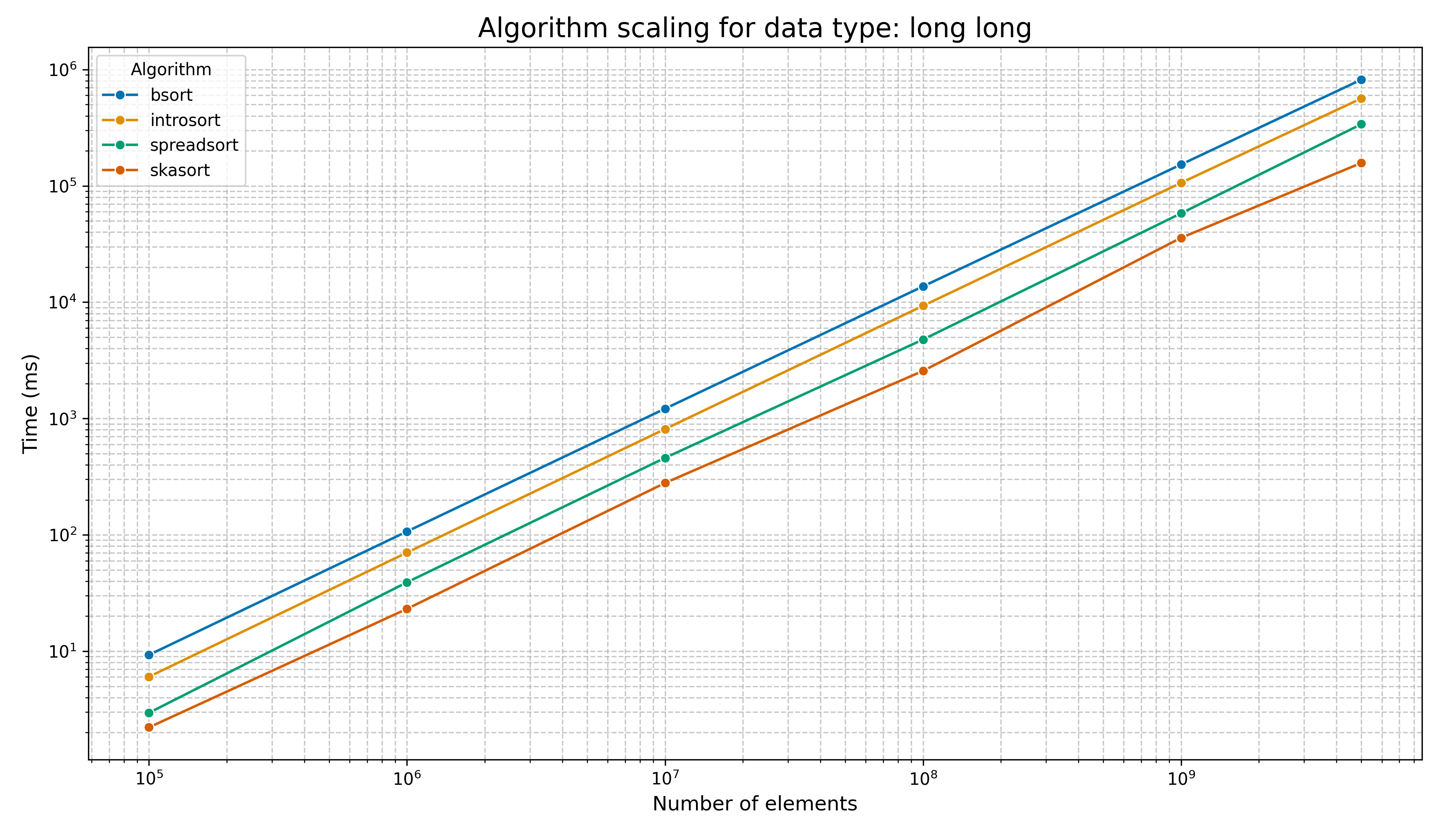}
	\end{subfigure}
	\caption{Performance comparison for integer types \texttt{char}, \texttt{short}, \texttt{int}, and \texttt{long long}.}
	\label{fig:comparison-by-ints}
\end{figure}

\begin{figure}[ht]
	\centering
	\newcommand{\plotwidth}{0.95\columnwidth}
	\begin{subfigure}{\plotwidth}
		\includegraphics[width=\linewidth]{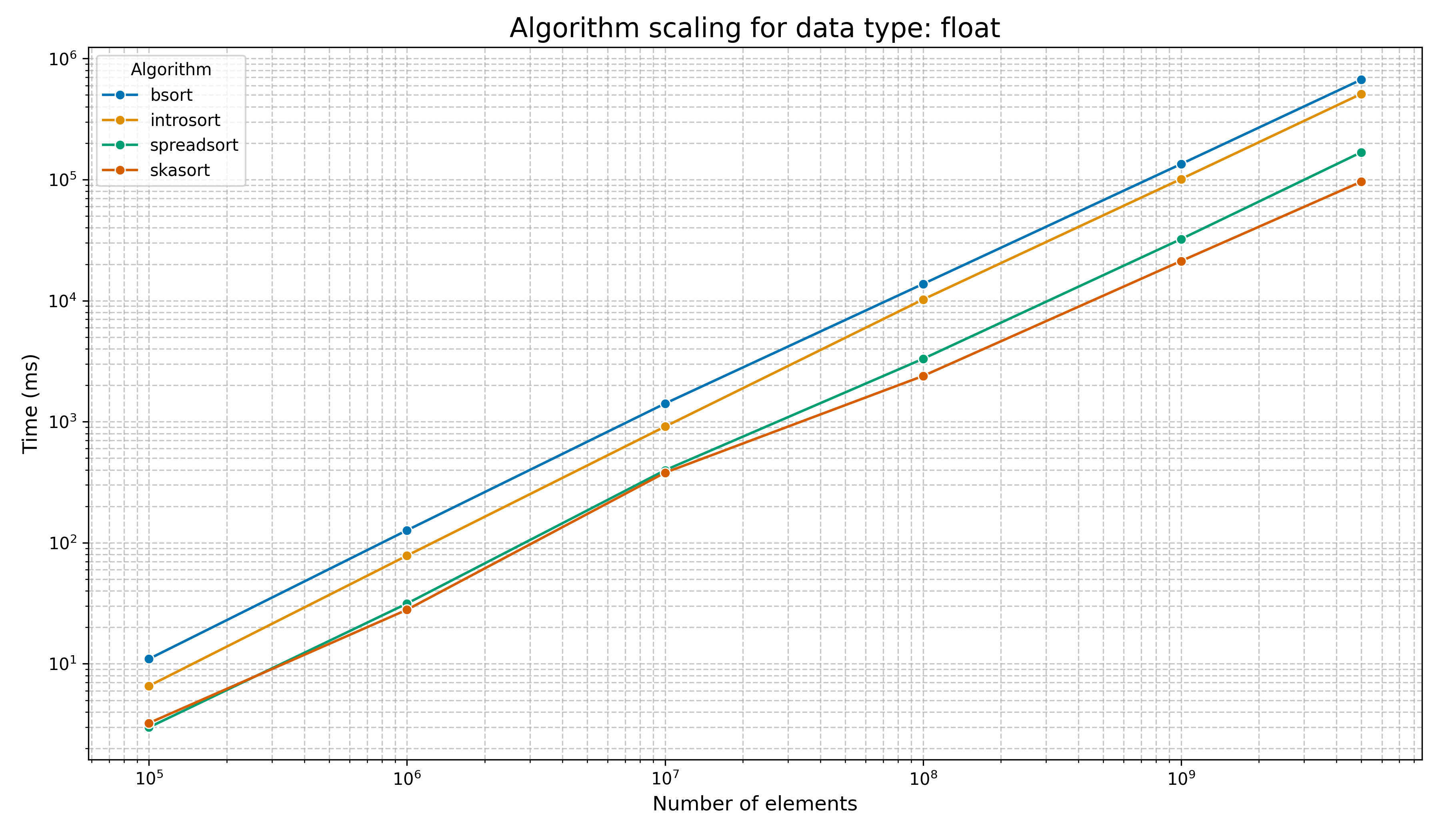}
	\end{subfigure}
	\begin{subfigure}{\plotwidth}
		\includegraphics[width=\linewidth]{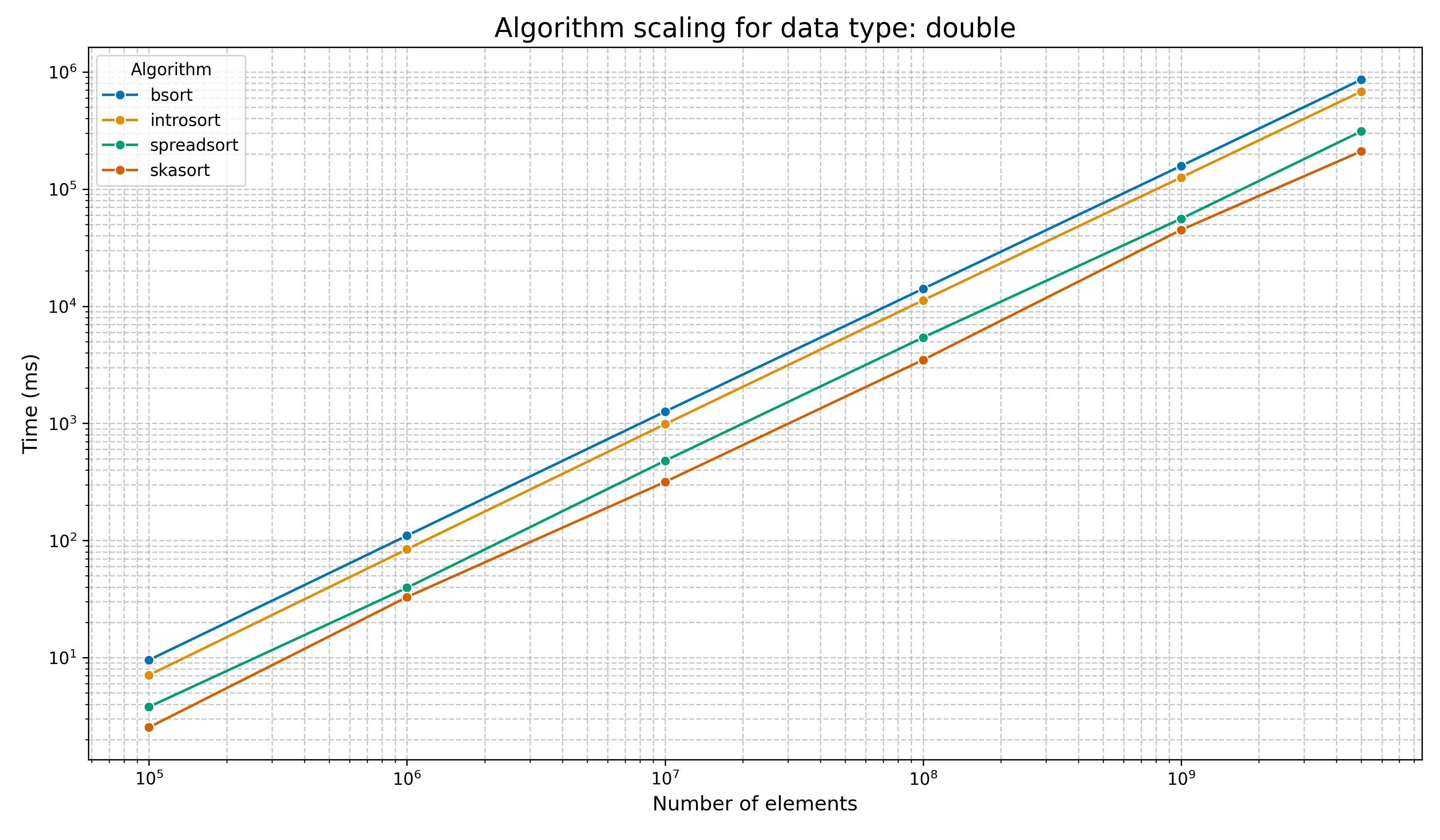}
	\end{subfigure}
	\caption{Performance comparison for floating-point types \texttt{float} and \texttt{double}.}
	\label{fig:comparison-by-floats}
\end{figure}

From Figures~\ref{fig:comparison-by-ints} and~\ref{fig:comparison-by-floats} it can be seen that $bsort$ demonstrates consistent linear scaling. As predicted by the analysis of its runtime in the previous section, the execution time grows linearly with the number of elements $n$. The performance gap between bsort and comparison-based introsort remains relatively stable as $n$ increases, confirming that the bitwise partitioning overhead does not degrade disproportionately for large datasets. 

Results also show that when the bit-depth is shallow, $bsort$ consistently outperforms introsort, but when it is deep, its relative performance decreases. This is a direct consequence of the $w$ factor, as more recursive passes are required to settle the final order.


Theoretically, bsort possesses an asymptotic advantage over comparison-based algorithms in specific regimes, as its runtime behavior is bound by $O(wn)$ and that of efficient comparison-based algorithms usually is $O(n \log (n) )$, the relative efficiency can be characterized by the ratio $\lim\limits_{n\rightarrow\infty} \frac{w n}{n log(n)} = 0$ and $\lim\limits_{w \rightarrow 0} \frac{w n}{n log(n)} = 0$. These limits suggest that bsort should strictly outperform any comparison-based algorithm as the array size n grows sufficiently large or the word size w becomes sufficiently small. Indeed, our results for the 8-bit \texttt{char} data type confirm this.

However, this theoretical superiority fails to materialize due to fundamental microarchitectural constraints. To isolate the specific bottlenecks impeding bsort's performance, the algorithm was profiled using Linux's \texttt{perf} tool on the same machine afore benchmarks were executed.

The empirical analysis identifies three primary causes for bsort's performance degradation: branch unpredictability, stack pollution and instructional volume. First, the algorithm relies on bitwise conditional checks inside the partitioning loop. On random data, these checks are inherently unpredictable (approaching a 50\% miss rate), causing frequent pipeline flushes. This is clearly shown on Figure~\ref{fig:branch-misses}.

Second, bsort employs a rigid recursive structure that imposes a hidden penalty: register pressure and stack pollution. Unlike hybrid algorithms that switch to iterative, low-overhead methods for small datasets, bsort incurs the full cost of recursive function calls down to the final bit. This repeated context switching forces frequent register spilling and consumes valuable L1 cache lines for stack management rather than payload data. Figure~\ref{fig:cache-misses} illustrates this inefficiency, showing bsort consistently incurs a high L1 D-cache miss rate across all data types.

\begin{figure}[ht]
	\centering
	\newcommand{\plotwidth}{0.95\columnwidth}
	\begin{subfigure}{\plotwidth}
		\includegraphics[width=\linewidth]{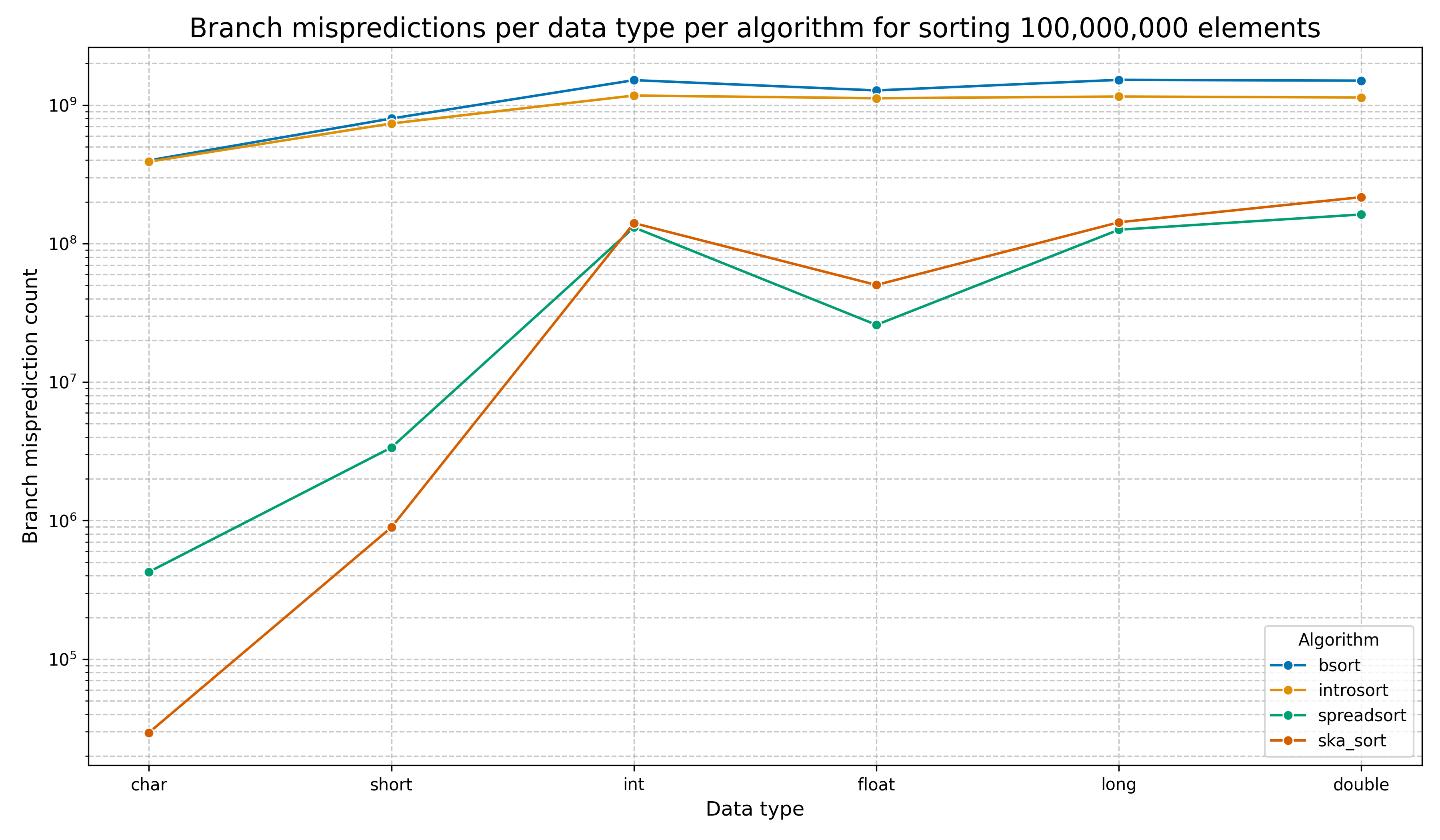}
	\end{subfigure}
	\caption{Comparison of branch misprediction count}
	\label{fig:branch-misses}
\end{figure}

\begin{figure}[ht]
	\centering
	\newcommand{\plotwidth}{0.95\columnwidth}
	\begin{subfigure}{\plotwidth}
		\includegraphics[width=\linewidth]{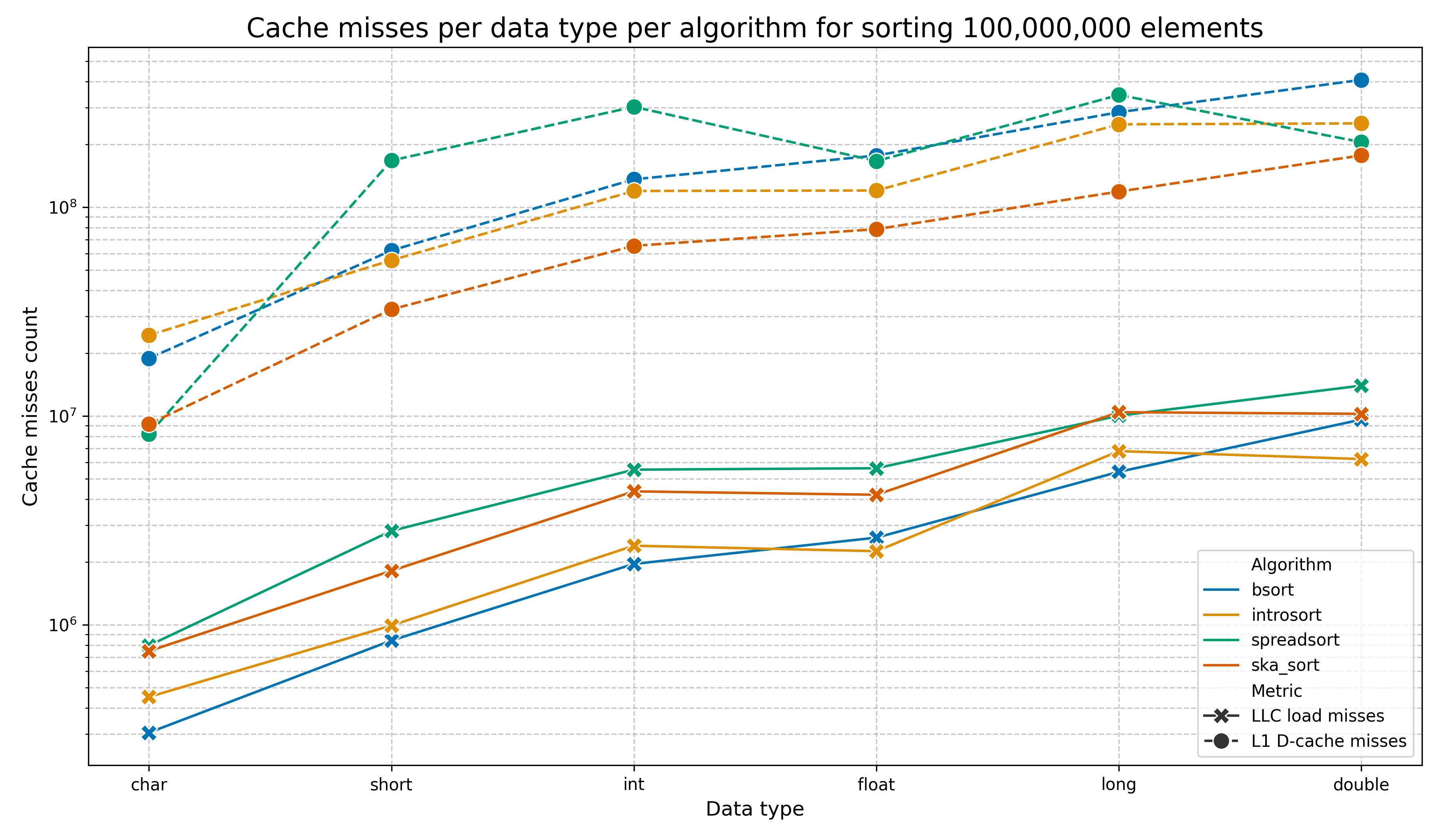}
	\end{subfigure}
	\caption{Comparison of cache misses count}
	\label{fig:cache-misses}
\end{figure}

The same recursive structure makes bsort traverse the whole array $w$ times, whereas other algorithms like introsort typically require only $O(log(n))$ times. Consequently, the sheer number of instructions required to complete the recursion overwhelms the processor as it executes significantly more instructions than other algorithms, e.g., 64 complete array scans for 64-bit integers in bsort vs. $\approx$ 26 for $10^8$ elements in introsort. This behavior is clearly shown in Figure~\ref{fig:cycle-count}.

\begin{figure}[ht]
	\centering
	\newcommand{\plotwidth}{0.95\columnwidth}
	\begin{subfigure}{\plotwidth}
		\includegraphics[width=\linewidth]{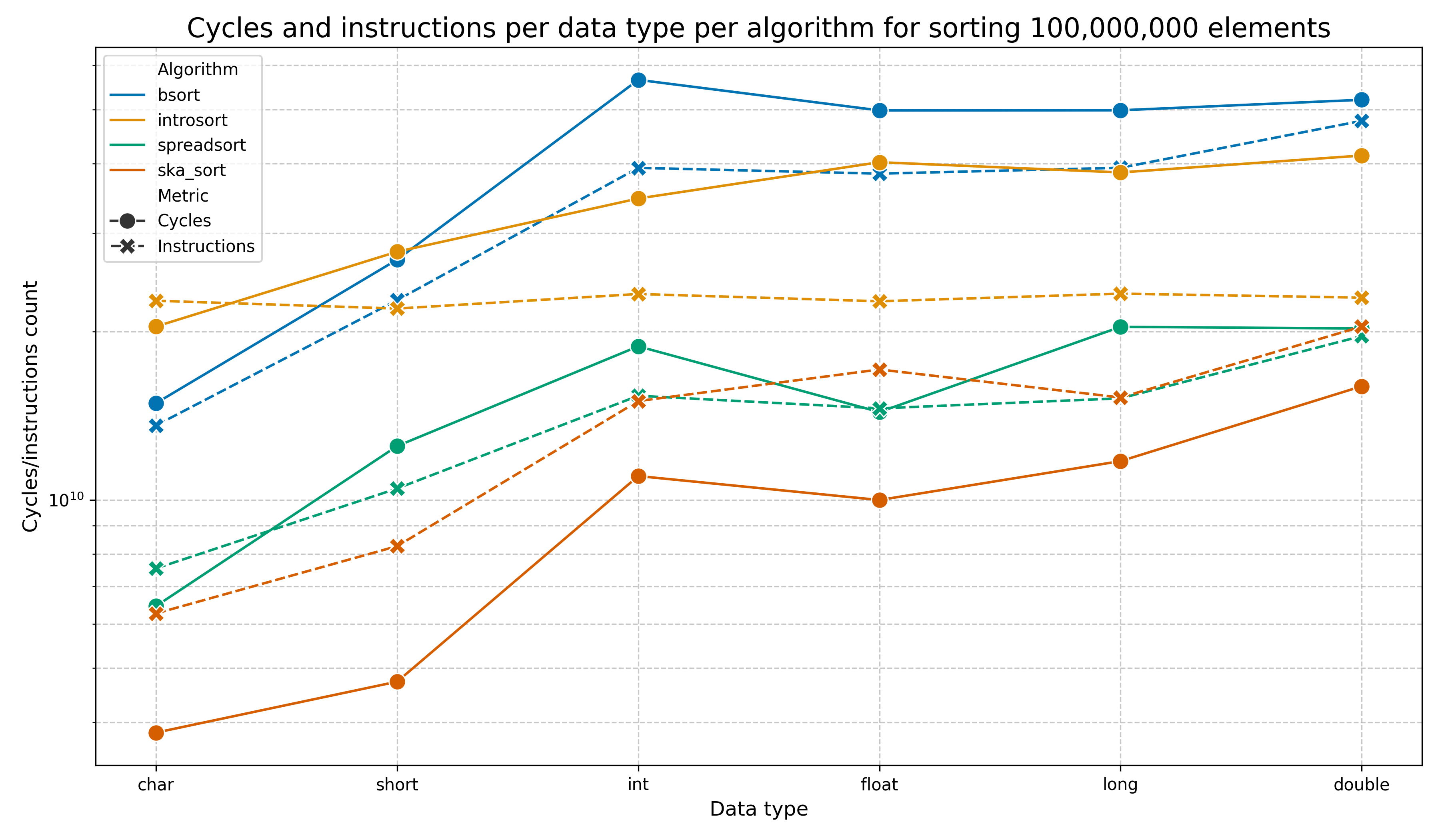}
	\end{subfigure}
	\caption{Comparison of CPU cycles and instructions count}
	\label{fig:cycle-count}
\end{figure}

The performance gap between bsort and other algorithms also stems from a fundamental difference in algorithmic architecture. As detailed in Table~\ref{table:algorithm-comparison}, state-of-the-art implementations like introsort, spreadsort, and ska\_sort are \textit{hybrid} algorithms; they adaptively switch strategies depending on partition size or data distribution to maximize hardware efficiency~\cite{introsort,skasort,spreadsort}. By contrast, \texttt{bsort} relies exclusively on the same bitwise partitioning procedure down to the finest granularity.

Remarkably, despite the lack of hybrid optimizations and the inherent instruction overhead, \texttt{bsort} demonstrates competitive performance. In shallow bit-depth regimes (e.g., 8-bit), it statistically outperforms the hybrid, cache-optimized introsort. This suggests that while the current monolithic implementation faces microarchitectural bottlenecks for large word sizes, the core bitwise mechanism possesses significant latent potential for future architectural optimization.

\section{Future enhancements}

While this work establishes the foundational mechanics and performance characteristics of $bsort$, several avenues for optimization remain to be explored. To address the performance gaps observed in larger word-size regimes, future iterations of the algorithm could be transformed into a hybrid implementation. By adaptively switching to a non-recursive, cache-friendly subroutine once partitions reach an optimal size, the $O(nw)$ recursion depth penalty could be significantly mitigated. Furthermore, performance could be enhanced through the integration of Single Instruction, Multiple Data (SIMD) instructions to parallelize bit-masking operations, branchless partitioning like conditional move instructions to prevent branch mispredictions, and exploiting Instruction-Level Parallelism (ILP) to overlap partitioning tasks.

These advanced architectural optimizations are currently outside the scope of this paper, particularly as the hardware utilized for this study does not provide native support for modern SIMD extensions. Consequently, the implementation and evaluation of these enhancements are reserved for subsequent research.

\section{Conclusion}

This paper presents $bsort$, an in-place, non-comparison-based sorting algorithm for signed and unsigned integers and floating-point values in $O(wn)$ time and $O(w)$ space. It was demonstrated through formal proof and experimental analysis that $bsort$ is particularly effective for small word-size regimes, such as 8-bit integers. However, it was also identified that for larger word sizes, the algorithm's performance is significantly hindered by its lack of a hybrid architecture and poor cache locality compared to state-of-the-art implementations. While $bsort$ provides a memory-efficient alternative for specific data types, future work is required to bridge the gap between its theoretical potential and practical high-performance execution.

\section*{Acknowledgment}

I want to express my deep gratitude to Kim García for her support and belief in me, without which many personal achievements, including this paper, would not have been possible.

\bibliographystyle{IEEEtran}
\bibliography{bsort.bib}

\end{document}